\documentclass[a4paper,magyar,english]{article}
\usepackage{mathptmx}
\usepackage[T1]{fontenc}
\usepackage[latin2,latin9]{inputenc}
\usepackage{units}
\usepackage{amsthm}
\usepackage{amsmath}
\usepackage{amssymb}
\usepackage{graphicx}
\usepackage{esint}

\makeatletter

\pdfpageheight\paperheight
\pdfpagewidth\paperwidth

\theoremstyle{plain}
\newtheorem{thm}{\protect\theoremname}
\newenvironment{lyxlist}[1]
{\begin{list}{}
{\settowidth{\labelwidth}{#1}
 \setlength{\leftmargin}{\labelwidth}
 \addtolength{\leftmargin}{\labelsep}
 }}
{\end{list}}

\allowdisplaybreaks[0]

\newcounter{meg}
\newcounter{pont}

\makeatother

\usepackage{babel}
  \addto\captionsenglish{\renewcommand{\theoremname}{Theorem}}
  \addto\captionsmagyar{\renewcommand{\theoremname}{\inputencoding{latin2}Tétel}}
\providecommand{\theoremname}{Theorem}

\begin{document}

\title{How to Move an Electromagnetic Field?}

\author{Márton Gömöri {\normalsize and} László E. Szabó\emph{\small }\\
\emph{\small Department of Logic, Institute of Philosophy}\\
\emph{\small Eötvös University, Budapest}\\
}

\date{~}
\maketitle
\begin{abstract}
The special relativity principle presupposes that the states of the
physical system concerned can be meaningfully characterized, at least
locally, as such in which the system is at rest or in motion with
some velocity relative to an arbitrary frame of reference. In the
first part of the paper we show that electrodynamic systems, in general,
do not satisfy this condition. In the second part of the paper we
argue that exatly the same condition serves as a necessary condition
for the persistence of an extended physical object. As a consequence,
we argue, electromagnetic field strengths cannot be the individuating
properties of electromagnetic field---contrary to the standard realistic
interpretation of CED. In other words, CED is ontologically incomplete.
\end{abstract}

\section{Introduction\label{sec:Introduction}}

The problem we address in this paper is on the border-line between
physics and metaphysics. We begin with the observation that the special
relativity principle (RP) is about the comparison of the behaviors
of physical systems in different states of inertial \emph{motion}
relative to an arbitrary inertial frame of reference. Therefore, it
is a minimal requirement for the RP to be a meaningful statement that
the states of the system in question must be meaningfully characterized
as such in which the system as a whole is at rest or in motion with
some velocity relative to an arbitrary frame of reference. Thus, to
apply the RP to classical electrodynamics (CED), it has to be meaningfully
formulated when an electrodynamic system---charged particles plus
electromagnetic field---is at rest or in motion relative to an inertial
frame of reference. In the first part of the paper we formulate a
minimal condition a solution of the Maxwell--Lorentz equations must
satisfy in order to describe such an electrodynamic configuration.
Then we prove that the solutions of the Maxwell--Lorentz equations,
in general, do not satisfy these conditions.

In the second part of the paper, we discuss the conceptual relationship
between the problem of motion and the problem of persistence. We argue
that persistence presupposes---zero or non-zero---velocity. One can
formulate a necessary condition for the persistence of an object,
in terms of its individuating properties. This condition implies that
the object must be\emph{ }in motion with some instantaneous velocity;
or, in case of an extended object, its local parts must be in motion
with some local and instantaneous velocities. At this point the problem
of persistence connects to the problem discussed in the first part
of the paper. As it is proved in Section~\ref{sec:Is-relativity-principle},
electromagnetic field does not satisfy this condition. Therefore,
we conclude, electromagnetic field cannot be regarded as a real physical
entity persisting in space and time; or, the field strengths cannot
be regarded as fundamental quantities individuating electromagnetic
field, that is, electrodynamics cannot be regarded as an ontologically
complete description of electromagnetic phenomena.

\section{The RP Is about the Behaviors of Physical Systems in Different States
of Motion}

The RP is one of the fundamental principles which must be satisfied
by all laws of physics describing any physical phenomena. Without
entering into the more technical formulation of the principle (see
e.g. Gömöri and Szabó~2013), we would like to focus on one particular
aspect, which is already clearly there in Galileo's first formulation:
\begin{quote}
Shut yourself up with some friend in the main cabin below decks on
some large ship, and have with you there some flies, butterflies,
and other small flying animals. Have a large bowl of water with some
fish in it; hang up a bottle that empties drop by drop into a wide
vessel beneath it. With the ship standing still, observe carefully
how the little animals fly with equal speed to all sides of the cabin.
The fish swim indifferently in all directions; the drops fall into
the vessel beneath; and, in throwing something to your friend, you
need throw it no more strongly in one direction than another, the
distances being equal; jumping with your feet together, you pass equal
spaces in every direction. When you have observed all these things
carefully (though doubtless when the ship is standing still everything
must happen in this way), have the ship proceed with any speed you
like, so long as the motion is uniform and not fluctuating this way
and that. You will discover not the least change in all the effects
named, nor could you tell from any of them whether the ship was moving
or standing still. In jumping, you will pass on the floor the same
spaces as before, nor will you make larger jumps toward the stern
than toward the prow even though the ship is moving quite rapidly,
despite the fact that during the time that you are in the air the
floor under you will be going in a direction opposite to your jump.
In throwing something to your companion, you will need no more force
to get it to him whether he is in the direction of the bow or the
stern, with yourself situated opposite. The droplets will fall as
before into the vessel beneath without dropping toward the stern,
although while the drops are in the air the ship runs many spans.
The fish in their water will swim toward the front of their bowl with
no more effort than toward the back, and will go with equal ease to
bait placed anywhere around the edges of the bowl. Finally the butterflies
and flies will continue their flights indifferently toward every side,
nor will it ever happen that they are concentrated toward the stern,
as if tired out from keeping up with the course of the ship, from
which they will have been separated during long intervals by keeping
themselves in the air. And if smoke is made by burning some incense,
it will be seen going up in the form of a little cloud, remaining
still and moving no more toward one side than the other. The cause
of all these correspondences of effects is the fact that \emph{the
ship's motion is common to all the things contained in it} {[}italics
added{]}, and to the air also. That is why I said you should be below
decks; for if this took place above in the open air, which would not
follow the course of the ship, more or less noticeable differences
would be seen in some of the effects noted. (Galilei 1953, 187)
\end{quote}
What is important for our present concern is that the principle is
about the comparison of the behaviors of physical systems---flies,
butterflies, fishes, droplets, smoke---\emph{in different states of
inertial motion} relative to an arbitrary inertial frame of reference.
In Brown's words:
\begin{quote}
The principle compares the outcome of relevant processes inside the
cabin under different states of inertial motion of the cabin relative
to the shore. It is simply assumed by Galileo that the same initial
conditions in the cabin can always be reproduced. What gives the relativity
principle empirical content is the fact that the differing states
of motion of the cabin are clearly distinguishable relative to the
earth\textquoteright{}s rest frame. (Brown 2005, 34)
\end{quote}
The RP describes the relationship between two situations: one is in
which the system, as a whole, is at rest relative to one inertial
frame, say $K$, the other is in which the system shows the similar
behavior, but being in a collective motion relative to $K$, co-moving
with some $K'$. In other words, the RP assigns to each solution $F$
of the physical equations, stipulated to describe the situation in
which the system is co-moving as a whole with inertial frame $K$,
another solution $M_{\mathbf{V}}(F)$, describing the similar behavior
of the same system when it is, as a whole, co-moving with inertial
frame\emph{ $K'$}, that is, when it is in a collective motion with
velocity $\mathbf{V}$ relative to $K$, where $\mathbf{V}$ is the
velocity of $K'$ relative to $K$. And it asserts that the solution
$M_{\mathbf{V}}(F)$, expressed in the primed variables of $K'$,
has exactly the same form as $F$ in the original variables of $K$.

Consequently, the following is a minimal requirement for the RP to
be a meaningful statement:

\paragraph*{Minimal Requirement for the RP (MR)}

\emph{The states of the system in question---described by the solutions
$F$---must be meaningfully characterized as such in which the system
as a whole is at rest or in motion with some velocity relative to
an arbitrary frame of reference.\medskip{}
}

Let us show a well-known electrodynamic example in which a particles
+ electromagnetic field system satisfies this condition. Consider
one single charged particle moving with constant velocity $\mathbf{V}=\left(V,0,0\right)$
relative to $K$ and the coupled stationary electromagnetic field
(Jackson~1999, 661): 
\begin{equation}
M_{\mathbf{V}}(F)\,\,\,\left\{ \begin{alignedat}{1}E_{x}(x,y,z,t) & =\frac{qX_{0}}{\left(X_{0}^{2}+\left(y-y_{0}\right)^{2}+\left(z-z_{0}\right)^{2}\right)^{\nicefrac{3}{2}}}\\
E_{y}(x,y,z,t) & =\frac{\gamma q\left(y-y_{0}\right)}{\left(X_{0}^{2}+\left(y-y_{0}\right)^{2}+\left(z-z_{0}\right)^{2}\right)^{\nicefrac{3}{2}}}\\
E_{z}(x,y,z,t) & =\frac{\gamma q\left(z-z_{0}\right)}{\left(X_{0}^{2}+\left(y-y_{0}\right)^{2}+\left(z-z_{0}\right)^{2}\right)^{\nicefrac{3}{2}}}\\
B_{x}(x,y,z,t) & =0\\
B_{y}(x,y,z,t) & =-c^{-2}VE_{z}\\
B_{z}(x,y,z,t) & =c^{-2}VE_{y}\\
\varrho(x,y,z,t) & =q\delta\left(x-(x_{0}+Vt)\right)\delta\left(y-y_{0}\right)\delta\left(z-z_{0}\right)
\end{alignedat}
\right.\label{eq:Coulomb-mozgo-1-1}
\end{equation}
where $(x_{0},y_{0},z_{0})$ is the initial position of the particle
at $t=0$, $X_{0}=\gamma\left(x-\left(x_{0}+Vt\right)\right)$ and
$\gamma=\left(1-\frac{V^{2}}{c^{2}}\right)^{-\frac{1}{2}}$. In this
case, it is no problem to characterize the particle + electromagnetic
field system as such which is, as a whole,\emph{ }in motion with velocity
$\mathbf{V}$ relative to $K$; as\emph{ }the electromagnetic field
is in collective motion with the point charge of velocity $\mathbf{V}$
(Fig.~\ref{fig:The-stationary-field-1}) in the following sense:%
\footnote{It must be pointed out that velocity $\mathbf{V}$ conceptually differs
from the speed of light $c$. Basically, $c$ is a constant of nature
in the Maxwell--Lorentz equations, which can emerge in the solutions
of the equations; and, in some cases, it can be interpreted as the
velocity of propagation of changes in the electromagnetic field. For
example, in our case, the stationary field of a uniformly moving point
charge, in collective motion with velocity $\mathbf{V},$ can be constructed
from the superposition of retarded potentials, in which the retardation
is calculated with velocity $c$; nevertheless, the two velocities
are different concepts. To illustrate the difference, consider the
fields of a charge at rest (\ref{eq:Coulomb field-1}), and in motion
(\ref{eq:Coulomb-mozgo-1-1}). The speed of light $c$ plays the same
role in both cases. Both fields can be constructed from the superposition
of retarded potentials in which the retardation is calculated with
velocity $c$. Also, in both cases, a small local perturbation in
the field configuration would propagate with velocity $c$. But still,
there is a consensus to say that the system described by (\ref{eq:Coulomb field-1})
is at rest while the one described by (\ref{eq:Coulomb-mozgo-1-1})
is moving with velocity $\mathbf{V}$ (together with $K'$, relative
to $K$.) A good analogy would be a Lorentz contracted moving rod:
$\mathbf{V}$ is the velocity of the rod, which differs from the speed
of sound in the rod. %
} 
\begin{figure}
\begin{centering}
\includegraphics[width=0.6\columnwidth]{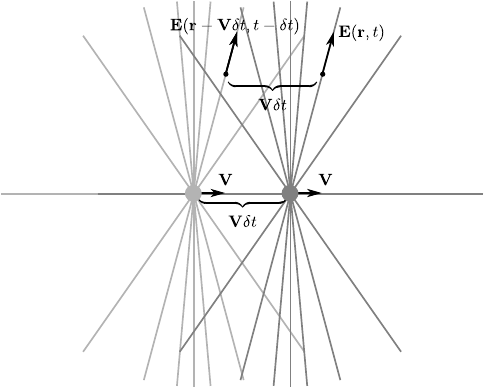}
\par\end{centering}

\caption{The stationary field of a uniformly moving point charge is in collective
motion together with the point charge \label{fig:The-stationary-field-1}}
\end{figure}
\begin{eqnarray}
\mathbf{E}(\mathbf{r},t) & = & \mathbf{E}(\mathbf{r}-\mathbf{V}\delta t,t-\delta t)\label{eq:mozgomezo-0-1}\\
\mathbf{B}(\mathbf{r},t) & = & \mathbf{B}(\mathbf{r}-\mathbf{V}\delta t,t-\delta t)\label{eq:mozgomezo-0-2}
\end{eqnarray}
that is,
\begin{eqnarray}
-\partial_{t}\mathbf{E}(\mathbf{r},t) & = & \mathsf{D}\mathbf{E}(\mathbf{r},t)\mathbf{V}\label{eq:mozgomezo1}\\
-\partial_{t}\mathbf{B}(\mathbf{r},t) & = & \mathsf{D}\mathbf{B}(\mathbf{r},t)\mathbf{V}\label{eq:mozgomezo2}
\end{eqnarray}
where $\mathsf{D}\mathbf{E}(\mathbf{r},t)$ and $\mathsf{D}\mathbf{B}(\mathbf{r},t)$
denote the spatial derivative operators (Jacobians for variables $x,y$
and $z$); that is, in components:
\begin{eqnarray}
-\partial_{t}E_{x}(\mathbf{r},t) & = & V_{x}\partial_{x}E_{x}(\mathbf{r},t)+V_{y}\partial_{y}E_{x}(\mathbf{r},t)+V_{z}\partial_{z}E_{x}(\mathbf{r},t)\label{eq:elsokomponens}\\
-\partial_{t}E_{y}(\mathbf{r},t) & = & V_{x}\partial_{x}E_{y}(\mathbf{r},t)+V_{y}\partial_{y}E_{y}(\mathbf{r},t)+V_{z}\partial_{z}E_{y}(\mathbf{r},t)\label{eq:masodikkomponens}\\
 & \vdots\nonumber \\
-\partial_{t}B_{z}(\mathbf{r},t) & = & V_{x}\partial_{x}B_{z}(\mathbf{r},t)+V_{y}\partial_{y}B_{z}(\mathbf{r},t)+V_{z}\partial_{z}B_{z}(\mathbf{r},t)\label{eq:utolsokomponens}
\end{eqnarray}

The uniformly moving point charge + electromagnetic field system not
only satisfies condition MR, but it satisfies the RP: Formula (\ref{eq:Coulomb-mozgo-1-1})
with $\mathbf{V}=0$ describes the static field of the particle when
they are at rest in $K$ :
\begin{equation}
F\,\,\,\left\{ \begin{alignedat}{1}E_{x}(x,y,z,t) & =\frac{q\left(x-x_{0}\right)}{\left(\left(x-x_{0}\right)^{2}+\left(y-y_{0}\right)^{2}+\left(z-z_{0}\right)^{2}\right)^{\nicefrac{3}{2}}}\\
E_{y}(x,y,z,t) & =\frac{q\left(y-y_{0}\right)}{\left(\left(x-x_{0}\right)^{2}+\left(y-y_{0}\right)^{2}+\left(z-z_{0}\right)^{2}\right)^{\nicefrac{3}{2}}}\\
E_{z}(x,y,z,t) & =\frac{q\left(z-z_{0}\right)}{\left(\left(x-x_{0}\right)^{2}+\left(y-y_{0}\right)^{2}+\left(z-z_{0}\right)^{2}\right)^{\nicefrac{3}{2}}}\\
B_{x}(x,y,z,t) & =0\\
B_{y}(x,y,z,t) & =0\\
B_{z}(x,y,z,t) & =0\\
\varrho(x,y,z,t) & =q\delta\left(x-x_{0}\right)\delta\left(y-y_{0}\right)\delta\left(z-z_{0}\right)
\end{alignedat}
\right.\label{eq:Coulomb field-1}
\end{equation}
By means of the Lorentz transformation rules one can express (\ref{eq:Coulomb-mozgo-1-1})
in terms of the `primed' variables of the co-moving reference frame
$K'$:
\begin{equation}
\begin{alignedat}{1}E'_{x}(x',y',z',t') & =\frac{q'\left(x'-x'_{0}\right)}{\left(\left(x'-x'_{0}\right)^{2}+\left(y'-y'_{0}\right)^{2}+\left(z'-z'_{0}\right)^{2}\right)^{\nicefrac{3}{2}}}\\
E'_{y}(x',y',z',t') & =\frac{q'\left(y'-y'_{0}\right)}{\left(\left(x'-x'_{0}\right)^{2}+\left(y'-y'_{0}\right)^{2}+\left(z'-z'_{0}\right)^{2}\right)^{\nicefrac{3}{2}}}\\
E'_{z}(x',y',z',t') & =\frac{q'\left(z'-z'_{0}\right)}{\left(\left(x'-x'_{0}\right)^{2}+\left(y'-y'_{0}\right)^{2}+\left(z'-z'_{0}\right)^{2}\right)^{\nicefrac{3}{2}}}\\
B'_{x}(x',y',z',t') & =0\\
B'_{y}(x',y',z',t') & =0\\
B'_{z}(x',y',z',t') & =0\\
\varrho(x',y',z',t') & =q\delta\left(x'-x'_{0}\right)\delta\left(y'-y'_{0}\right)\delta\left(z'-z'_{0}\right)
\end{alignedat}
\label{eq:Coulomb-vesszos-1}
\end{equation}
and we find that the result is indeed of the same form as (\ref{eq:Coulomb field-1}). 

So, in this well-known particular textbook example the RP is meaningful
and satisfied. This picture is in complete accordance with the standard
realistic interpretation of electromagnetic field:
\begin{quote}
In the standard interpretation of the formalism, the field strengths
$\mathbf{B}$ and $\mathbf{E}$ are interpreted realistically: The
interaction between charged particles are mediated by the electromagnetic
field, which is ontologically on a par with charged particles and
the state of which is given by the values of the field strengths.
(Frisch~2005, 28) 
\end{quote}
In this example, the charged particle and the coupled electromagnetic
field constitute a physical system which---just like Galileo's flies,
butterflies, fishes, droplets, and smoke---can be subject to the RP.
The states $F$ and $M_{\mathbf{V}}(F)$ can be meaningfully characterized
as such in which both parts of the physical system, the particle and
the electromagnetic field, are at rest or in motion with some velocity
relative to an arbitrary frame of reference\emph{.} We will show,
however, that this is not the case in general.

\section{How to Understand the RP for a General Electrodynamic System?\label{sec:Is-relativity-principle}}

What meaning can be attached to the words ``a coupled particles +
electromagnetic field system is \emph{in collective} \emph{motion}
with velocity $\mathbf{V}$'' ($\mathbf{V}=0$ included) relative
to a reference frame $K$, in general? One might think, we can read
off the answer to this question from the above example. However, focusing
on the electromagnetic field, the partial differential equations (\ref{eq:mozgomezo1})--(\ref{eq:mozgomezo2})
imply that
\begin{eqnarray}
\mathbf{E}(\mathbf{r},t) & = & \mathbf{E}_{0}(\mathbf{r}-\mathbf{V}t)\label{eq:stac1}\\
\mathbf{B}(\mathbf{r},t) & = & \mathbf{B}_{0}(\mathbf{r}-\mathbf{V}t)\label{eq:stac2}
\end{eqnarray}
with some time-independent $\mathbf{E}_{0}(\mathbf{r})$ and $\mathbf{B}_{0}(\mathbf{r})$.
In other words, the field must be a stationary one, that is, a translation
of a static field with velocity $\mathbf{V}$. But, (\ref{eq:stac1})--(\ref{eq:stac2})
is certainly not the case for a general solution of the equations
of CED; the field is not necessarily translating with a collective
velocity. The behavior of the field can be much more complex. Whatever
this complex behavior is, it is quite intuitive to assume that the
following general principle must hold:

\paragraph*{Mereological Principle of Motion (MPM)}

\emph{If an extended object as a whole is at rest or is in motion
with some velocity relative to an arbitrary reference frame $K$,
then all local parts of it are in motion with some local instantaneous
velocity $\mathbf{v}(\mathbf{r},t)$ relative to $K$.}

\medskip{}

Combining MPM with MR, we obtain the following:

\paragraph*{Local Minimal Requirement for the RP (LMR)}

\emph{The states of the extended physical system in question must
be meaningfully characterized as such in which all local parts of
the system are at rest or in motion with some local instantaneous
velocity relative to an arbitrary frame of reference.\medskip{}
}

Consequently, in case of electrodynamics, a straightforward minimal
requirement for the RP to be a meaningful statement is that (\ref{eq:mozgomezo-0-1})--(\ref{eq:mozgomezo-0-2})
must be satisfied at least \emph{locally} with some local and instantaneous
velocity $\mathbf{v}(\mathbf{r},t)$: it is quite natural to say that
the electromagnetic field at point $\mathbf{r}$ and time $t$ is\emph{
}moving\emph{ }with \emph{local} and \emph{instantaneous} velocity
$\mathbf{v}(\mathbf{r},t)$ if and only if
\begin{eqnarray}
\mathbf{E}(\mathbf{r},t) & = & \mathbf{E}\left(\mathbf{r}-\mathbf{v}(\mathbf{r},t)\delta t,t-\delta t\right)\label{eq:mozgomezo-0a-1}\\
\mathbf{B}(\mathbf{r},t) & = & \mathbf{B}\left(\mathbf{r}-\mathbf{v}(\mathbf{r},t)\delta t,t-\delta t\right)\label{eq:mozgomezo-0a-2}
\end{eqnarray}
are satisfied \emph{locally,} in an \emph{infinitesimally} small space
and time region at $(\mathbf{r},t)$, for infinitesimally small $\delta t$.
In other words, the equations (\ref{eq:mozgomezo1})--(\ref{eq:mozgomezo2})
must be satisfied \emph{locally} at point $(\mathbf{r},t)$ with a
local and instantaneous velocity $\mathbf{v}(\mathbf{r},t)$: 
\begin{eqnarray}
-\partial_{t}\mathbf{E}(\mathbf{r},t) & = & \mathsf{D}\mathbf{E}(\mathbf{r},t)\mathbf{v}(\mathbf{r},t)\label{eq:mozgomezo1a}\\
-\partial_{t}\mathbf{B}(\mathbf{r},t) & = & \mathsf{D}\mathbf{B}(\mathbf{r},t)\mathbf{v}(\mathbf{r},t)\label{eq:mozgomezo2a}
\end{eqnarray}

In other words, if the RP, as it is believed, applies to all situations
in electrodynamics, there must exist a local instantaneous velocity
field $\mathbf{v}(\mathbf{r},t)$ satisfying (\ref{eq:mozgomezo1a})--(\ref{eq:mozgomezo2a})
for all possible solutions of the following system of Maxwell--Lorentz
equations:

\begin{eqnarray}
\nabla\cdot\mathbf{E}\left(\mathbf{r},t\right) & = & \sum_{i=1}^{n}q^{i}\delta\left(\mathbf{r}-\mathbf{r}^{i}\left(t\right)\right)\label{eq:MLE1}\\
c^{2}\nabla\times\mathbf{B}\left(\mathbf{r},t\right)-\partial_{t}\mathbf{E}\left(\mathbf{r},t\right) & = & \sum_{i=1}^{n}q^{i}\delta\left(\mathbf{r}-\mathbf{r}^{i}\left(t\right)\right)\mathbf{v}^{i}\left(t\right)\label{eq:MLE2}\\
\nabla\cdot\mathbf{B}\left(\mathbf{r},t\right) & = & 0\label{eq:MLE3}\\
\nabla\times\mathbf{E}\left(\mathbf{r},t\right)+\partial_{t}\mathbf{B}\left(\mathbf{r},t\right) & = & 0\label{eq:MLE4}\\
m^{i}\gamma\left(\mathbf{v}^{i}\left(t\right)\right)\mathbf{a}^{i}(t) & = & q^{i}\left\{ \mathbf{E}\left(\mathbf{r}^{i}\left(t\right),t\right)+\mathbf{v}^{i}\left(t\right)\times\mathbf{B}\left(\mathbf{r}^{i}\left(t\right),t\right)\right.\,\,\,\nonumber \\
 &  & \left.-c^{-2}\mathbf{v}^{i}\left(t\right)\left(\mathbf{v}^{i}\left(t\right)\mathbf{\cdot E}\left(\mathbf{r}^{i}\left(t\right),t\right)\right)\right\} \label{eq:MLE5}\\
 &  & \,\,\,\,\,\,\,\,\,\,\,\,\,\,\,\,\,\,\,\,\,\,\,\,\,\,\,\,\,(i=1,2,\ldots,n)\nonumber 
\end{eqnarray}
where, $\gamma(\ldots)=\left(1-\frac{(\ldots)^{2}}{c^{2}}\right)^{-\frac{1}{2}}$,
$q^{i}$ is the electric charge and $m^{i}$ is the rest mass\emph{
}of the $i$-th particle. That is, substituting an arbitrary solution%
\footnote{Without entering into the details, it must be noted that the Maxwell--Lorentz
equations (\ref{eq:MLE1})--(\ref{eq:MLE5}), exactly in this form,
have \emph{no} solution. The reason is that the field is singular
at precisely the points where the coupling happens: on the trajectories
of the particles. The generally accepted answer to this problem is
that the real source densities are some ``smoothed out'' Dirac deltas,
determined by the physical laws of the internal worlds of the particles---which
are, supposedly, outside of the scope of CED. With this explanation,
for the sake of simplicity we leave the Dirac deltas in the equations.
Since our considerations here focuses on the electromagnetic field,
satisfying the four Maxwell equations, we must only assume that there
is a coupled dynamics---approximately described by equations (\ref{eq:MLE1})--(\ref{eq:MLE5})---and
that it constitutes an initial value problem. In fact, Theorem~\ref{thm:There-is-a}
could be stated in a weaker form, by leaving the concrete form and
dynamics of the source densities unspecified.%
} of (\ref{eq:MLE1})--(\ref{eq:MLE5}) into (\ref{eq:mozgomezo1a})--(\ref{eq:mozgomezo2a}),
the overdetermined system of equations must have a solution for $\mathbf{v}(\mathbf{r},t)$. 

However, one encounters the following difficulty: 
\begin{thm}
\label{thm:There-is-a}There is a dense subset of solutions of the
coupled Maxwell--Lorentz equations (\ref{eq:MLE1})--(\ref{eq:MLE5})
for which there cannot exist a local instantaneous velocity field
$\mathbf{v}(\mathbf{r},t)$ satisfying (\ref{eq:mozgomezo1a})--(\ref{eq:mozgomezo2a}).\end{thm}
\begin{proof}
The proof is almost trivial for a locus $(\mathbf{r},t)$ where there
is a charged point particle. However, in order to avoid the eventual
difficulties concerning the physical interpretation, we are providing
a proof for a point $(\mathbf{r}_{*},t_{*})$ where there is assumed
no source at all. 

Consider a solution $\left(\mathbf{r}^{1}\left(t\right),\mathbf{r}^{2}\left(t\right),\ldots,\mathbf{r}^{n}\left(t\right),\mathbf{E}(\mathbf{r},t),\mathbf{B}(\mathbf{r},t)\right)$
of the coupled Maxwell--Lorentz equations (\ref{eq:MLE1})--(\ref{eq:MLE5}),
which satisfies (\ref{eq:mozgomezo1a})--(\ref{eq:mozgomezo2a}).
At point $(\mathbf{r}_{*},t_{*})$, the following equations hold:
\begin{eqnarray}
-\partial_{t}\mathbf{E}(\mathbf{r}_{*},t_{*}) & = & \mathsf{D}\mathbf{E}(\mathbf{r}_{*},t_{*})\mathbf{v}(\mathbf{r}_{*},t_{*})\label{eq:mozgomezo1c}\\
-\partial_{t}\mathbf{B}(\mathbf{r}_{*},t_{*}) & = & \mathsf{D}\mathbf{B}(\mathbf{r}_{*},t_{*})\mathbf{v}(\mathbf{r}_{*},t_{*})\label{eq:mozgomezo2c}\\
\partial_{t}\mathbf{E}(\mathbf{r}_{*},t_{*}) & = & c^{2}\nabla\times\mathbf{B}(\mathbf{r}_{*},t_{*})\label{eq:elsomaxwell}\\
-\partial_{t}\mathbf{B}(\mathbf{r}_{*},t_{*}) & = & \nabla\times\mathbf{E}(\mathbf{r}_{*},t_{*})\label{eq:masodikmaxwell}\\
\nabla\cdot\mathbf{E}(\mathbf{r}_{*},t_{*}) & = & 0\label{eq:dive}\\
\nabla\cdot\mathbf{B}(\mathbf{r}_{*},t_{*}) & = & 0\label{eq:divb}
\end{eqnarray}
Without loss of generality we can assume---at point $\mathbf{r}_{*}$
and time $t_{*}$---that operators $\mathsf{D}\mathbf{E}(\mathbf{r}_{*},t_{*})$
and $\mathsf{D}\mathbf{B}(\mathbf{r}_{*},t_{*})$ are invertible and
$v_{z}(\mathbf{r}_{*},t_{*})\neq0$.

Now, consider a $3\times3$ matrix $J$ such that 
\begin{equation}
J=\left(\begin{array}{ccc}
\partial_{x}E_{x}(\mathbf{r}_{*},t_{*}) & J_{xy} & J_{xz}\\
\partial_{x}E_{y}(\mathbf{r}_{*},t_{*}) & \partial_{y}E_{y}(\mathbf{r}_{*},t_{*}) & \partial_{z}E_{y}(\mathbf{r}_{*},t_{*})\\
\partial_{x}E_{z}(\mathbf{r}_{*},t_{*}) & \partial_{y}E_{z}(\mathbf{r}_{*},t_{*}) & \partial_{z}E_{z}(\mathbf{r}_{*},t_{*})
\end{array}\right)\label{eq:jacobi}
\end{equation}
with

\begin{eqnarray}
J_{xy} & = & \partial_{y}E_{x}(\mathbf{r}_{*},t_{*})+\lambda\label{eq:jatek0}\\
J_{xz} & = & \partial_{z}E_{x}(\mathbf{r}_{*},t_{*})-\lambda\frac{v_{y}(\mathbf{r}_{*},t_{*})}{v_{z}(\mathbf{r}_{*},t_{*})}\label{eq:jatek}
\end{eqnarray}
by virtue of which
\begin{eqnarray}
J_{xy}v_{y}(\mathbf{r}_{*},t_{*})+J_{xz}v_{z}(\mathbf{r}_{*},t_{*}) & = & v_{y}(\mathbf{r}_{*},t_{*})\partial_{y}E_{x}(\mathbf{r}_{*},t_{*})\nonumber \\
 &  & +v_{z}(\mathbf{r}_{*},t_{*})\partial_{z}E_{x}(\mathbf{r}_{*},t_{*})
\end{eqnarray}
Therefore, $J\mathbf{v}(\mathbf{r}_{*},t_{*})=\mathsf{D}\mathbf{E}(\mathbf{r}_{*},t_{*})\mathbf{v}(\mathbf{r}_{*},t_{*})$.
There always exists a vector field $\mathbf{E}_{\lambda}^{\#}(\mathbf{r})$
such that its Jacobian matrix at point $\mathbf{r}_{*}$ is equal
to $J$. Obviously, from (\ref{eq:dive}) and (\ref{eq:jacobi}),
$\nabla\cdot\mathbf{E}_{\lambda}^{\#}(\mathbf{r}_{*})=0$. Therefore,
there exists a solution of the Maxwell--Lorentz equations, such that
the electric and magnetic fields $\mathbf{E}_{\lambda}(\mathbf{r},t)$
and $\mathbf{B}_{\lambda}(\mathbf{r},t)$ satisfy the following conditions:%
\footnote{$\mathbf{E}_{\lambda}^{\#}(\mathbf{r})$ and $\mathbf{B}_{\lambda}(\mathbf{r},t_{*})$
can be regarded as the initial configurations at time $t_{*}$; we
do not need to specify a particular choice of initial values for the
sources.%
} 
\begin{eqnarray}
\mathbf{E}_{\lambda}(\mathbf{r},t_{*}) & = & \mathbf{E}_{\lambda}^{\#}(\mathbf{r})\\
\mathbf{B}_{\lambda}(\mathbf{r},t_{*}) & = & \mathbf{B}(\mathbf{r},t_{*})
\end{eqnarray}
At $(\mathbf{r}_{*},t_{*})$, such a solution obviously satisfies
the following equations: 
\begin{eqnarray}
\partial_{t}\mathbf{E}_{\lambda}(\mathbf{r}_{*},t_{*}) & = & c^{2}\nabla\times\mathbf{B}(\mathbf{r}_{*},t_{*})\label{eq:elsomaxwell*}\\
-\partial_{t}\mathbf{B}_{\lambda}(\mathbf{r}_{*},t_{*}) & = & \nabla\times\mathbf{E}_{\lambda}^{\#}(\mathbf{r}_{*})\label{eq:masodikmaxwell*}
\end{eqnarray}
therefore 
\begin{equation}
\partial_{t}\mathbf{E}_{\lambda}(\mathbf{r}_{*},t_{*})=\partial_{t}\mathbf{E}(\mathbf{r}_{*},t_{*})\label{eq:idoderivaltakegyenlok}
\end{equation}

As a little reflection shows, if $\mathsf{D}\mathbf{E}_{\lambda}^{\#}(\mathbf{r}_{*})$,
that is $J$, happened to be not invertible, then one can choose a
\emph{smaller} $\lambda$ such that $\mathsf{D}\mathbf{E}_{\lambda}^{\#}(\mathbf{r}_{*})$
becomes invertible (due to the fact that $\mathsf{D}\mathbf{E}(\mathbf{r}_{*},t_{*})$
is invertible), and, at the same time, 
\begin{equation}
\nabla\times\mathbf{E}_{\lambda}^{\#}(\mathbf{r}_{*})\neq\nabla\times\mathbf{E}(\mathbf{r}_{*},t_{*})\label{eq:rotnem}
\end{equation}
Consequently, from \eqref{eq:idoderivaltakegyenlok} , \eqref{eq:jatek}
and \eqref{eq:mozgomezo1c} we have
\begin{equation}
-\partial_{t}\mathbf{E}_{\lambda}(\mathbf{r}_{*},t_{*})=\mathsf{D}\mathbf{E}_{\lambda}(\mathbf{r}_{*},t_{*})\mathbf{v}(\mathbf{r}_{*},t_{*})=\mathsf{D}\mathbf{E}_{\lambda}^{\#}(\mathbf{r}_{*})\mathbf{v}(\mathbf{r}_{*},t_{*})
\end{equation}
and $\mathbf{v}(\mathbf{r}_{*},t_{*})$ is uniquely determined by
this equation. On the other hand, from \eqref{eq:masodikmaxwell*}
and \eqref{eq:rotnem} we have
\begin{equation}
-\partial_{t}\mathbf{B}_{\lambda}(\mathbf{r}_{*},t_{*})\neq\mathsf{D}\mathbf{B}_{\lambda}(\mathbf{r}_{*},t_{*})\mathbf{v}(\mathbf{r}_{*},t_{*})=\mathsf{D}\mathbf{B}(\mathbf{r}_{*},t_{*})\mathbf{v}(\mathbf{r}_{*},t_{*})
\end{equation}
because $\mathsf{D}\mathbf{B}(\mathbf{r}_{*},t_{*})$ is invertible,
too. That is, for $\mathbf{E}_{\lambda}(\mathbf{r},t)$ and $\mathbf{B}_{\lambda}(\mathbf{r},t)$
there is no local and instantaneous velocity at point $\mathbf{r}_{*}$
and time $t_{*}$. 

At the same time, $\lambda$ can be arbitrary small, and 
\begin{eqnarray}
\lim_{\lambda\rightarrow0}\mathbf{E}_{\lambda}(\mathbf{r},t) & = & \mathbf{E}(\mathbf{r},t)\\
\lim_{\lambda\rightarrow0}\mathbf{B}_{\lambda}(\mathbf{r},t) & = & \mathbf{B}(\mathbf{r},t)
\end{eqnarray}
Therefore solution $\left(\mathbf{r}_{\lambda}^{1}\left(t\right),\mathbf{r}_{\lambda}^{2}\left(t\right),\ldots,\mathbf{r}_{\lambda}^{n}\left(t\right),\mathbf{E}_{\lambda}(\mathbf{r},t),\mathbf{B}_{\lambda}(\mathbf{r},t)\right)$
can fall into an arbitrary small neighborhood of $\left(\mathbf{r}^{1}\left(t\right),\mathbf{r}^{2}\left(t\right),\ldots,\mathbf{r}^{n}\left(t\right),\mathbf{E}(\mathbf{r},t),\mathbf{B}(\mathbf{r},t)\right)$.%
\footnote{Notice that our investigation has been concerned with the general
laws of Maxwell--Lorentz electrodynamics of a coupled particles +
electromagnetic field system. The proof of the theorem was essentially
based on the presumption that all solutions of the Maxwell--Lorentz
equations, determined by \emph{any} initial state of the particles
+ electromagnetic field system, corresponded to physically possible
configurations of the electromagnetic field. It is sometimes claimed,
however, that the solutions must be restricted by the so called retardation
condition, according to which all physically admissible field configurations
must be generated from the retarded potentials belonging to some pre-histories
of the charged particles (Jánossy 1971, p. 171; Frisch 2005, p.~145).
There is no obvious answer to the question of how Theorem~\ref{thm:There-is-a}
is altered under such additional condition.%
}
\end{proof}
Thus, the meaning of the concept of ``electromagnetic field moving\emph{
}with a local instantaneous velocity $\mathbf{v}(\mathbf{r},t)$ at
point $\mathbf{r}$ and time $t$'', that we obtained by a straightforward
generalization of the example of the stationary field of a uniformly
moving charge, is untenable. We do not see other available rational
meaning of this concept. Such a concept, on the other hand, would
be a necessary conceptual plugin to the RP. In any event, lacking
a better suggestion, we must conclude that the RP is a statement which
is meaningless for a general electrodynamic situation.

\section{No Persistence without Motion}

There is a long debate in contemporary metaphysics whether and in
what sense instantaneous velocity can be regarded as an intrinsic
property of an object at a given moment of time (Butterfield~2006;
Arntzenius~2000; Tooley~1988; Hawley~2001, 76--80; Sider~2001,
34--35). There seems to be, however, a consensus that
\begin{quote}
{[}$\ldots${]} the notion of velocity presupposes the persistence
of the object concerned. For average velocity is a quotient, whose
numerator must be the distance traversed by the given persisting object:
otherwise you could give me a superluminal velocity by dividing the
distance between me and the Sun by a time less than eight minutes.
So presumably, average velocity\textquoteright{}s limit, instantaneous
velocity, also presupposes persistence. (Butterfield 2005, 257)
\end{quote}
In this section we argue that the opposite is also true: the notion
of persistence requires the existence of instantaneous velocity. 

It is common to all theories of persistence---endurantism vs. perdurantism---that
a persisting entity needs to have some package of individuating properties,
in terms of which one can express that two things in two different
spatiotemporal regions are identical, or at least constitute different
spatiotemporal parts of the same entity. Butterfield writes: 
\begin{quote}
I believe that {[}the criteria of identity{]} are largely independent
of the endurantism--perdurantism debate; and in particular, that endurantism
and perdurantism {[}...{]} face some common questions about criteria
of identity, and can often give the same, or similar, answers to them.
{[}...{]} {[}A{]}ll parties need to provide criteria of identity for
objects, presumably invoking the usual notions of qualitative similarity
and-or causation (Butterfield 2005, 248--289)
\end{quote}
Without loss of generality we may assume that each of these individuating
properties can be characterized as such that a certain quantity $f_{i}$
takes a certain value. Consider a primitive example: the redness of
the ball in Fig.~\ref{fig:ball1} can be characterized as such that
the wavelength of light reflected from the instantaneous surface of
the ball is around 650 nm. Or, more abstractly, just imagine a quantity
the spatiotemporal distribution of which takes value $1$ in a region
where redness is instantiated---for example, on the locus of the ball---and
takes value $0$ otherwise. 

\begin{figure}
\begin{centering}
\includegraphics[width=0.75\columnwidth]{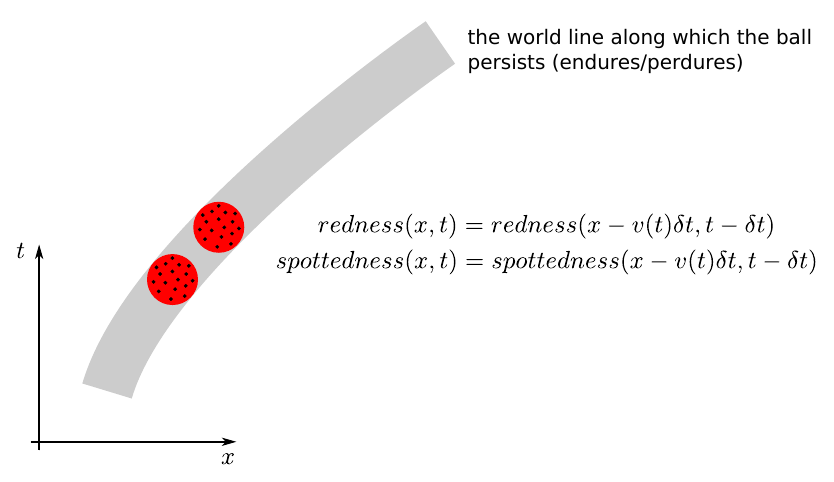}
\par\end{centering}

\caption{A ball is individuated by its redness, spottedness, etc. \label{fig:ball1}}
\end{figure}

Now, in order to express the fact of persistence, consider a given
$n$-tuple of individuating quantities $\left\{ f_{i}\right\} _{i=1}^{n}$
that is supposed to trace out the trajectory or spacetime tube along
which the entity persists. The different theories of persistence disagree
in the actual content of the package $\left\{ f_{i}\right\} _{i=1}^{n}$,
these differences are not important from the point of view of our
present concern. The following necessary condition is however common
to all intuitions: 
\begin{eqnarray}
f_{i}(\mathbf{r},t) & = & f_{i}(\mathbf{r}-\mathbf{v}\left(t\right)\delta t,t-\delta t)\label{eq:persistence_pointlike}\\
 &  & \,\,\,\,\,\,\,\,\,\,\,\,\,\,(i=1,2,\ldots,n)\nonumber 
\end{eqnarray}
for all $(\mathbf{r},t)$ where the object is present, at least for
a small, infinitesimal, interval of time $\delta t$ (Fig.~\ref{fig:ball1}),
with some instantaneous velocity $\mathbf{v}\left(t\right)$. Without
loss of generality we may assume that all functions in $\left\{ f_{i}\right\} _{i=1}^{n}$
are smooth (if not, we can approximate them by smooth functions).
Expressing \eqref{eq:persistence_pointlike} in a differential form,
we have%
\footnote{For the sake of simplicity we may assume that all $f_{i}$ are scalar
functions, and $\mathsf{D}f_{i}$ is simply $\mathsf{grad}\, f_{i}$.%
}
\begin{eqnarray}
-\partial_{t}f_{i}(\mathbf{r},t) & = & \mathsf{D}f_{i}(\mathbf{r},t)\mathbf{v}(t)\label{eq:equation_of_persistence_pointlike}\\
 &  & \,\,\,\,\,\,\,(i=1,2,\ldots,n)\nonumber 
\end{eqnarray}
In other words, the entity is \emph{in motion} with some instantaneous
velocities $\mathbf{v}(t)$. Let us call \eqref{eq:equation_of_persistence_pointlike}
the \emph{equations of persistence}.

So far we considered the situation when the persistence can be formulated
in terms of individuating quantities $\left\{ f_{i}\right\} _{i=1}^{n}$
characterizing the entity in question \emph{as a whole}. Generally,
however, this is not necessarily the case. An extended object may
persist, even if its holistic properties do not satisfy equations
(\ref{eq:equation_of_persistence_pointlike}). Following however the
same intuition by which we formulated the Mereological Principle of
Motion, we propose the following thesis:

\paragraph*{Mereological Principle of Persistence (MPP)}

\emph{If an extended object, as a whole, persists, then its all local
parts persist.}

\medskip{}

Accordingly, the persistence of an extended object requires the following
condition for the spatial distributions: 
\begin{eqnarray}
f_{i}(\mathbf{r},t) & = & f_{i}(\mathbf{r}-\mathbf{v}\left(\mathbf{r},t\right)\delta t,t-\delta t)\\
 &  & \,\,\,\,\,\,\,\,\,\,\,\,\,\,\,\,\,\,\,\,(i=1,2,\ldots,n)\nonumber 
\end{eqnarray}
or 
\begin{eqnarray}
-\partial_{t}f_{i}(\mathbf{r},t) & = & \mathsf{D}f_{i}(\mathbf{r},t)\mathbf{v}(\mathbf{r},t)\label{eq:equation_of_persistence}\\
 &  & \,\,\,\,\,\,\,\,\,\,(i=1,2,\ldots,n)\nonumber 
\end{eqnarray}
for all $(\mathbf{r},t)$ where the extended object is present; where
$\mathbf{v}(\mathbf{r},t)$ is a local and instantaneous velocity
field characterizing the \emph{motion} of the local part of the extended
entity at the spatiotemporal locus $(\mathbf{r},t)$ (Fig~\ref{fig:local_balls}).
Let us call \eqref{eq:equation_of_persistence} the \emph{equations
of persistence for an extended object.}

\begin{figure}
\begin{centering}
\includegraphics[width=0.6\columnwidth]{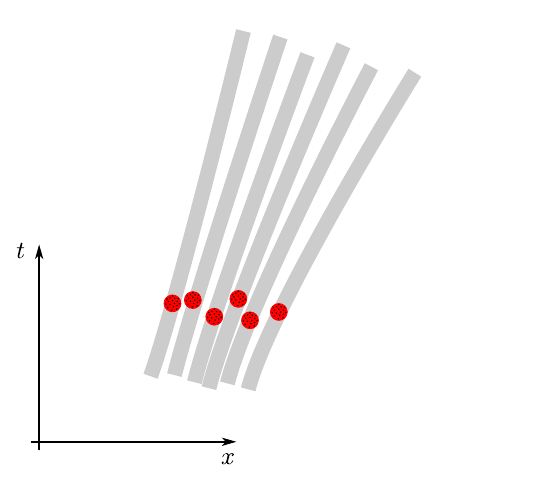}
\par\end{centering}

\caption{Persistence of an extended object requires the persitence of its local
parts\label{fig:local_balls}}
\end{figure}

\bigskip{}

\section{The Ontological Incompleteness of CED}

As we have seen in Theorem~\ref{thm:There-is-a}, the distributions
of the two fundamental electrodynamic field strengths, $\mathbf{E}(\mathbf{r},t)$
and $\mathbf{B}(\mathbf{r},t)$, do not satisfy the equations of persistence
(\ref{eq:equation_of_persistence}). Therefore, the electromagnetic
field individuated by the field strengths cannot be regarded as a
persisting physical object; in other words, electromagnetic field
cannot be regarded as being a real physical entity existing in space
and time. This seems to contradict to the usual realistic interpretation
of CED. 

If electromagnetic field is a real entity persisting in space and
time, then it cannot be individuated by the field strengths. That
is, there must exist some quantities other than the field strengths,
perhaps outside of the scope of CED, individuating the local parts
of electromagnetic field. This suggests that CED is an ontologically
incomplete theory.

How to conceive properties, different from the field strengths, which
are capable of individuating the electromagnetic field? One might
think of them as some ``finer'', more fundamental, properties of
the field, not only individuating it as a persisting extended object,
but also determining the values of the field strengths. However, the
following easily verifiable theorem shows that this determination
cannot be so simple:
\begin{thm}
Let $\left\{ f_{i}\right\} _{i=1}^{n}$ be a package of quantities
for which there exist a local instantaneous velocity field $\mathbf{v}(\mathbf{r},t)$
satisfying the equations of persistence \eqref{eq:equation_of_persistence}\emph{
}in a given spacetime region. If a quantity $\Phi$ is a function
of the quantities \textup{$f_{1},f_{2},\ldots,f_{n}$} in the following
form: 
\[
\Phi(\mathbf{r},t)=\Phi\left(f_{1}(\mathbf{r},t),f_{2}(\mathbf{r},t),...,f_{n}(\mathbf{r},t)\right)
\]
then $\Phi$ also obeys the equation of persistence 
\begin{eqnarray*}
-\partial_{t}\Phi(\mathbf{r},t) & = & \mathsf{D}\Phi(\mathbf{r},t)\mathbf{v}(\mathbf{r},t)
\end{eqnarray*}
with the same local instantaneous velocity field $\mathbf{v}(\mathbf{r},t)$,
within the same spacetime region.
\end{thm}
Therefore, $\mathbf{E}(\mathbf{r},t)$ and $\mathbf{B}(\mathbf{r},t)$
cannot supervene pointwise upon some more fundamental individuating
quantities satisfying the persistence equations. However, they might
supervene in some non-local sense. For example, imagine that $\mathbf{E}(\mathbf{r},t)$
and $\mathbf{B}(\mathbf{r},t)$ provide only a course-grained characterization
of the field, but there exist some more fundamental fields $\mathbf{e}(\mathbf{r},t)$
and $\mathbf{b}(\mathbf{r},t)$, such that
\begin{align*}
\mathbf{E}(\mathbf{r},t) & =\underset{\Omega}{\int}\mathbf{e}\left(\mathbf{r}',t'\right)d^{4}(\mathbf{r},t)\\
\mathbf{B}(\mathbf{r},t) & =\underset{\Omega}{\int}\mathbf{b}\left(\mathbf{r}',t\right)d^{4}(\mathbf{r},t)
\end{align*}
where $\Omega$ is a neighbourhood of $(\mathbf{r},t)$ (Fig.~\ref{fig:nonlocal_superrvenience}).
In this case, the more fundamental quantities $\mathbf{e}(\mathbf{r},t)$
and $\mathbf{b}(\mathbf{r},t)$ may satisfy the equations of persistence,
while $\mathbf{E}(\mathbf{r},t)$ and $\mathbf{B}(\mathbf{r},t)$,
supervening on $\mathbf{e}(\mathbf{r},t)$ and $\mathbf{b}(\mathbf{r},t)$,
may not.

\begin{figure}
\begin{centering}
\includegraphics[width=0.75\columnwidth]{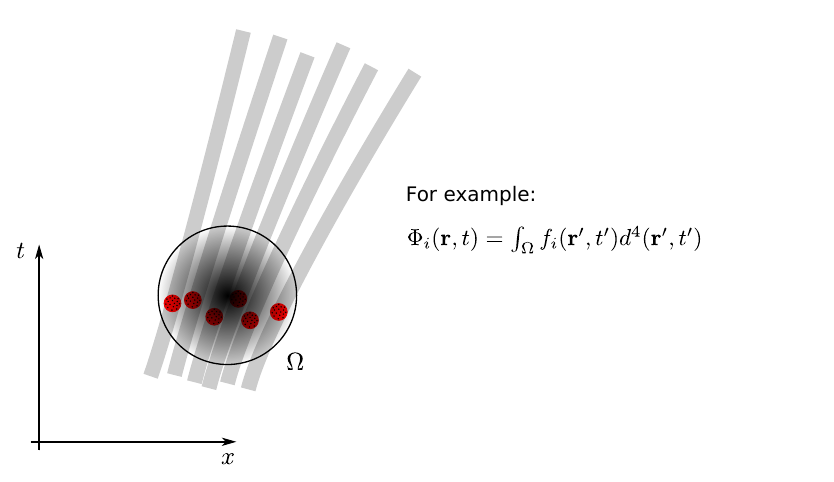}
\par\end{centering}

\caption{A non-local form of supervenience \label{fig:nonlocal_superrvenience}}
\end{figure}

\section*{Acknowledgment}

The research was partly supported by the OTKA Foundation, No.~K100715.

\section*{}

\section*{References}

\noindent Arntzenius,~Frank.~2000. ``Are There Really Instantaneous
Velocities?'' \emph{The Monist} 83:187--208.

~

\noindent Brown,~Harvey~R. 2005. \emph{Physical Relativity -- Space-Time
Structure from a Dynamical Perspectiv}e. Oxford, NY: Oxford University
Press.

~

\noindent Butterfield,~Jeremy. 2005. ``On the Persistence of Particles'',
\emph{Foundations of Physics} 35:233--269. 

~

\noindent Butterfield,~Jeremy. 2006. ``The Rotating Discs Argument
Defeated'', \emph{British Journal for the Philosophy of Science}
57:1--45.

~

\noindent Frisch,~Mathias. 2005. \emph{Inconsistency, Asymmetry,
and Non-Locality}, Oxford: Oxford University Press. 

~

\noindent Galilei, Galileo. 1953. \emph{Dialogue concerning the Two
Chief World Systems, Ptolemaic \& Copernican}, Berkeley: University
of California Press.

~

\noindent Gömöri,~Márton,~and~László~E.~Szabó. 2013. ``Formal
Statement of the Special Principle of Relativity'', \emph{Synthese},
DOI: 10.1007/s11229-013-0374-1

~

\noindent Hawley,~Katherine. 2001. \emph{How Things Persist}, Oxford:
Oxford University Press.

~

\noindent Jackson,~John~David. 1999. \emph{Classical Electrodynamics
(Third edition).} Hoboken, NJ: John Wiley \& Sons.

~

\noindent Jánossy,~Lajos. 1971. \emph{Theory of Relativity Based
On Physical Reality}, Budapest: \inputencoding{latin2}\foreignlanguage{magyar}{Akadémiai
Kiadó}\inputencoding{latin9}.\emph{}

~

\noindent Sider,~Theodore. 2001. \emph{Four-Dimensionalism}, Oxford:
Oxford University Press.

~

\noindent Tooley,~Michael. 1988. ``In Defence of the Existence of
States of Motion'', \emph{Philosophical Topics} 16:225--254. 
\begin{lyxlist}{00.00.0000}
\item [{}]~
\item [{}]~\end{lyxlist}

\end{document}